\documentclass[lnicst]{svmultln}
\usepackage{makeidx}  

\usepackage{color, array, colortbl}
\usepackage{graphicx}
\usepackage{subfigure}
\usepackage{amsfonts}
\usepackage{amsmath}
\usepackage{booktabs}
\usepackage{url}

\usepackage{etoolbox}
\newtoggle{paper}
\toggletrue{paper}

\newcommand{\cancel}[1]{}
\usepackage{hyperref}

\spnewtheorem*{xproof}{}{\itshape}{\rmfamily}

\renewenvironment{proof}[1][\proofname]
 {\renewcommand\xproofname{#1}\xproof}
 {\endxproof}

\def\inline#1:{\par\vskip 7pt\noindent{\bf #1:}\hskip 10pt}
\def\STRH{\STRTGY_{\text{\tiny $\H$}}}
\def\STRT{\STRTGY_{\text{\tiny $\T$}}}
\def\STRU{\STRTGY_{\text{\tiny $\U$}}}

\def\CC{\mathrm{x}}
\def\PROCESS{\Psi}
\def\cE{{\cal E}}

\def\Prb{\mathbb{P}}
\def\Expct{\mathbb{E}}
\def\Rnum{\mathbb{R}}

\def\rbratio{\alpha}

\def\at{\rbratio_t}
\def\NN{{\tt N}}
\def\MM{{\tt M}}
\def\XX{{\tt X}}

\def\RED{{\tt R}}
\def\BLUE{{\tt B}}

\def\RRED{\text{\tiny \tt R}}
\def\BBLUE{\text{\tiny \tt B}}
\def\rhoRR{{\rho_{\text{\tiny \tt R}}}}
\def\rhoBB{{\rho_{\text{\tiny \tt B}}}}

\def\pRR{{P_{\text{\tiny \tt R}\text{\tiny \tt R}}}}
\def\pBB{{P_{\text{\tiny \tt B}\text{\tiny \tt B}}}}

\def\H{\text{\tt H}}
\def\T{\text{\tt T}}
\def\U{\text{\tt U}}
\def\dR{\Delta_\RED}
\def\dB{\Delta_\BLUE}
\newcommand{\STRTGY}{\pi}
\def\bpa{\mathrm{BPA}}
\def\eh{\mathrm{EH}}

\newcommand{\card}[1]{\lvert #1\rvert}

\iftoggle{paper}{
\setlength{\textheight}{8.7in}
\setlength{\textwidth}{5.6in}
\setlength{\evensidemargin}{0.7cm}
\setlength{\oddsidemargin}{0.7cm}

\parskip=0.08in
}
{}

\begin{document}
\mainmatter              
\title{Assortative Mixing Equilibria in\\ Social Network Games\thanks{Supported in part by a grant of the Israel Science Foundation (1549/13)}}

\titlerunning{Assortative Mixing Equilibria}  
\iftoggle{paper}{}{\vspace{-3cm}}
\author{Chen Avin\inst{1} \and Hadassa Daltrophe\inst{1} \and Zvi Lotker\inst{1} \and David Peleg\inst{2}}
\authorrunning{Avin, Daltrophe, Lotker \& Peleg}   
%
%
\tocauthor{Hadassa Daltrophe}%
\institute{Ben Gurion University of the Negev, Beer Sheva, Israel\\
\email{avin@cse.bgu.ac.il,hd@cs.bgu.ac.il,zvilo@bgu.ac.il}\\
       \and
Weizmann Institute of Science, Rehovot, Israel\\
\email{david.peleg@weizmann.ac.il}
 }

\maketitle              
\iftoggle{paper}{}{\vspace{-10pt}}
\begin{abstract}        
It is known that individuals in social networks tend to exhibit 
{\em homophily} (a.k.a. {\em assortative mixing}) in their social ties, 
which implies that they prefer bonding with others of their own kind. 
But what are the reasons for this phenomenon? Is it that such relations are 
more convenient and easier to maintain? Or are there also 
some more tangible benefits to be gained from this collective behaviour?

The current work takes a game-theoretic perspective on this phenomenon, 
and studies the conditions under which different assortative mixing strategies 
lead to equilibrium in an evolving social network. 
We focus on a biased preferential attachment model where the strategy of 
each group (e.g., political or social minority) determines the level of bias 
of its members toward other group members and non-members.
Our first result is that if the utility function that the group attempts 
to maximize is the {\em degree centrality} of the group, interpreted as 
the sum of degrees of the group members in the network, then the only strategy 
achieving Nash equilibrium is a perfect homophily, which implies that 
cooperation with other groups is harmful to this utility function.
A second, and perhaps more surprising, result is that 
if a reward for inter-group cooperation is added to the utility function 
(e.g., externally enforced by an authority as a regulation), 
then there are only two possible equilibria, namely, {\em perfect homophily} 
or {\em perfect heterophily}, and it is possible to characterize 
their feasibility spaces. Interestingly, these results hold 
regardless of the minority-majority ratio in the population.

We believe that these results, as well as the game-theoretic perspective 
presented herein, may contribute to a better understanding of the forces 
that shape the groups and communities of our society.
\iftoggle{paper}{}{\vspace{-5pt}}
\iftoggle{paper}{}{\vspace{-5pt}}
\end{abstract}

\iftoggle{paper}{}{\vspace{-30pt}}
\section{Introduction}
\iftoggle{paper}{}{\vspace{-8pt}}
{\em Homophily} (lit. ``love of the same")~\cite{lazarsfeld1954friendship}, also known as {\em assortative mixing}~\cite{newman2003mixing}, is a prevalent and well documented phenomenon in social networks~\cite{mcpherson2001birds}; in making their social ties, people often prefer to connect with other individuals of similar characteristics, such as nationality, race, gender, age, religion, education or profession. 

Homophily
has many important consequences, both on the structure of the social network (e.g., the formation of communities) and on the behaviors and opportunities of participants in it, for example on the welfare of individuals \cite{jackson2008social} and on the diffusion patterns of information in the network \cite{jackson2013diffusion}.
It is therefore interesting to explore the reasons for this phenomenon.
Clearly, one natural reason is that relationship with similar individuals may be more convenient and easier to maintain. 
But are there also some more tangible benefits to be gained from this collective behaviour of sub-populations in the network?

To better understand homophily, we take a different perspective on this phenomenon and study it through a strategic, game-theoretic prism.
We investigate the conditions under which different assortative (and disassortative) mixing strategies lead to equilibrium in an \emph{evolving social network} game. 


To model the network evolution, we use a variant of the classical \emph{preferential attachment} model \cite{barabasi1999emergence}, which incorporates a 
heterogeneous population and assortative mixing patterns for the sub-populations. This model, known as \emph{biased preferential attachment} (BPA) ~\cite{avin2015homophily}, maintains the ``rich get richer" property, but additionally enables different mixing patterns (including perfect homophily and heterophily) between sub-populations, by using rejection sampling.


In this paper, we modify this model by turning it into a game. Each sub-population is represented as a player who can choose its mixing pattern as a strategy. The {\em utility function} (or {\em payoff}) of a player is a result of its population's (expected) properties in the BPA model. A {\em strategy profile} (describing the strategies of both players) attains a \emph{Nash equilibrium} for the game if no player  can do better by unilaterally changing its own strategy.

Obviously, the result of the game depends on the players' utility functions. In the current study we take an initial step and study two natural utility functions. In the first, we consider the payoff to be the total power of the group, that is, the sum of degrees of all group members.
In this case we prove that there is a unique stable Nash equilibrium which is the  \emph{perfect homophily} profile, namely, cooperation with other groups is harmful to this utility function. We stress that while there are other strategy profiles, like the unbaised profile, that guarantee the same 
total power to the groups, those profiles do not yield Nash equilibrium.

Since perfect homophily results in complete segregation of the sub-populations, we consider a second utility function based on a linear combination between the total power of the group and the number of \emph{cross-population} links (i.e., the size of the population cut). In particular, the utility is taken to be $\gamma$ times the total power of the group plus $1-\gamma$ times the population cut size, for some {\em weight factor} $0\le\gamma\le 1$. Such a utility can be viewed as a rule (or a law) imposed by a regulator to encourage cooperation between the two sub-populations. 
At a first glance, this utility seems to lead to different Nash equilibria for different $\gamma$ values. Somewhat surprisingly, we show that only two possible equilibria may emerge. For $\gamma > 1/2$, the \emph{perfect homophily} profile is the unique Nash equilibrium, and for $\gamma < 1/2$, the \emph{heterophily} profile is the unique Nash equilibrium. For $\gamma = 1/2$, both profiles yield a Nash equilibrium, but only the perfect homophily yields a stable equilibrium.
(Note, by the way, that all our results are independent of the ratio $r$ between the sizes of the two sub-populations.)

What may we learn from these results? A first, quite intuitive, lesson is that if the payoff includes benefits for heterophilic edges, then the game can move away from the perfect homophily equilibrium. But, within the natural utility function we study, if the game moves away from the homophily equilibrium, then it must reach a perfect heterophily equilibrium.  Both of these equilibria may appear to be too ``radical'' from a social capital perspective, which may find it desirable to maintain some balance in-between the two extremes, i.e., preserve the internal structure of both sub-populations as well as form significant 
cross-population links between the two sub-populations. This leaves us with some interesting follow-up research directions:
what `mechanism design' rules can a regulator employ in order to have a more fine-grained control on the equilibrium? what happens in a system with more than two sub-populations? how do the equilibria behave? We leave these questions for future work; we believe that taking the game theoretic perspective on evolving social network models for heterogenous populations is an important tool in understanding homophily, as shown in this initial model.

\iftoggle{paper}{}
{
Due to space limitations, we provide only an outline of our proofs. The interested reader is referred to~\url{https://goo.gl/h0Fegc} for details.
}




\begin{figure*}[t!]
\iftoggle{paper}{}{\vspace{-0.3cm}}
\begin{center}
	\begin{tabular}{ccc}
	   \includegraphics[width=.32\textwidth]{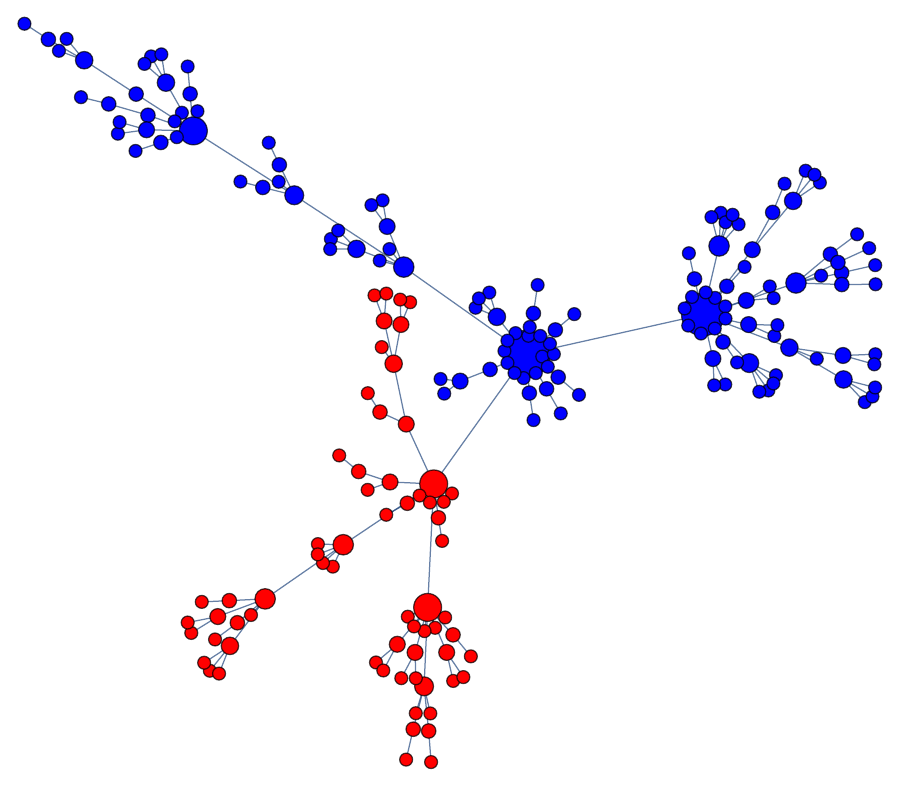}  &
	   \includegraphics[width=.35\textwidth]{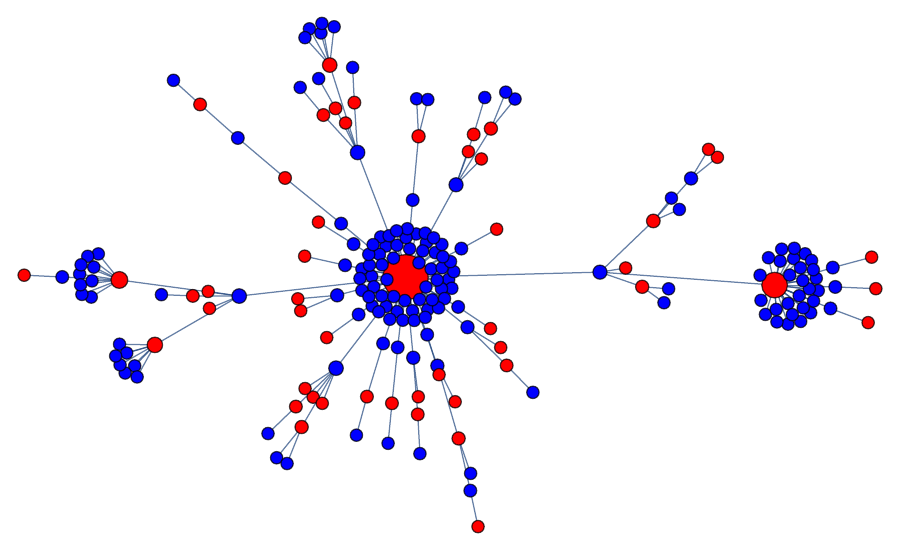} &
	    \includegraphics[width=.32\textwidth]{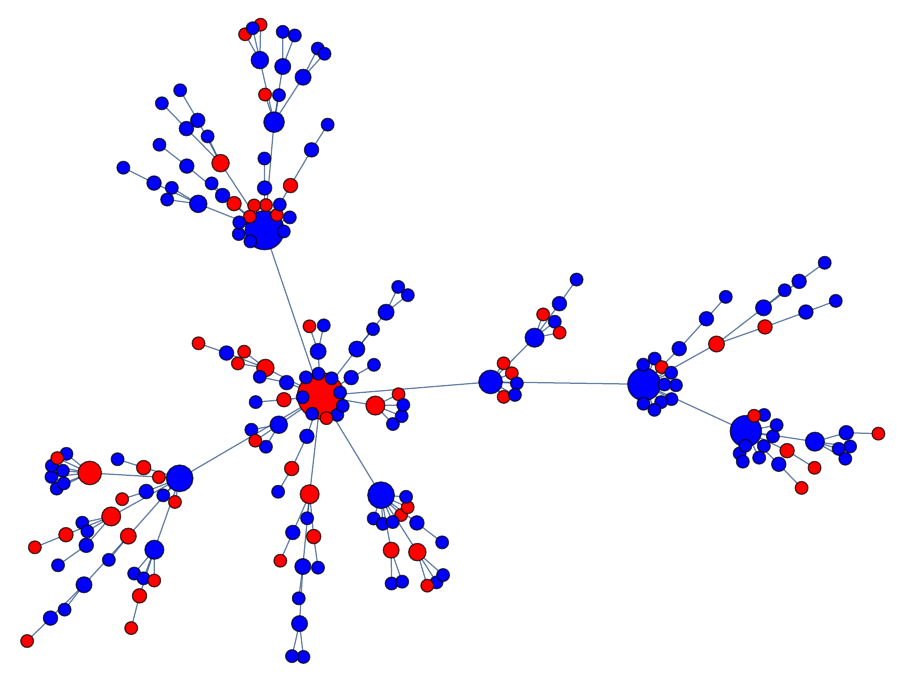} \\
		(a) $\STRH$ - homophily & (b) $\STRT$ - heterophily  &  (c) $\STRU$ - unbiased	
		\end{tabular}
	\caption{Examples of the Biased Preferential Attachment (BPA) model with various parameter settings. All examples depict a 200-vertex bi-populated network generated by our BPA model starting from a single edge connecting a blue and a red vertex and $30\%$ red nodes (with vertex size proportional to its degree).}         
	\label{fig:model}
\end{center}
\iftoggle{paper}{}{\vspace{-12pt}}
\end{figure*}

\iftoggle{paper}{}{\vspace{-11pt}}
\section{Related work}
\iftoggle{paper}{}{\vspace{-4pt}}

Game theory provides a natural framework for modeling selfish interests and the networks they generate~\cite{aumann1988endogenous,tardos2007network}. While many studies (see~\cite{jackson2005survey} for a comprehensive survey) focus on \textit{local} network formation games, others (e.g.,~\cite{chen2009network}) model the players as making \textit{global} structural decisions. In this paper we define a game that features a mixture of both local and global characteristics. This situation is close to cooperative games~\cite{bilbao2012cooperative}, where all the nodes of the same group have the same payment. However, the key idea of cooperative games is to choose which coalitions to form, whereas here the partition into groups is predefined. 

In this context, one should distinguish between \textit{network formation games} \cite{jackson2002evolution,jackson2005survey,tardos2007network} and \textit{evolving network games} (e.g.,~\cite{Bramoulle12}). The former involve a fixed set of nodes, with the connections between them changing over time. In contrast, in the evolving network model used herein, the nodes and edges are both dynamic, and new nodes join the network as it evolves over time.

Based on the assumption that people have tendency to copy the decisions of other people, we suggest a network construction process that follows the well known preferential attachment model~\cite{barabasi1999emergence} with an additional phase to incorporate the mixing parameter~\cite{avin2015homophily}. However, related studies in the economics literature examine different procedures to model the social network formation. The studies of~\cite{currarini2009economic,fu2012evolution} assume that individuals are randomly paired with other members of the population and then match assortatively. Another model, presented at~\cite{Bramoulle12}, suggests two-phase attachments. The nodes first choose their neighbors with a bias towards their own type and then make an unbiased choice of neighbors from among the neighbors of their biased neighbors. While the models of~\cite{jackson2002evolution,fu2012evolution} and others assume that a connecting edge between a pair of nodes is fixed by using bilateral agreement, in our model the matching choice is somewhat ambiguous. The rejection of a proposed connection can be interpreted as either decided by one of the parties unilaterally or accepted by a bilateral agreement. 

One of the main themes of this paper is studying the homophily phenomenon and its influence on minority-majority groups. McPherson et al.~\cite{mcpherson2001birds} give an overview of research on homophily and survey a variety of properties and how they lead to particular patterns in bonding. While some studies (e.g.,~\cite{currarini2009economic,di2015quantifying,avin2015homophily}) model homophily as ranging over a spectrum between perfect homophily and unbiased society, we have followed~\cite{Bramoulle12} and~\cite{fu2012evolution}, which also allow disassortative matching. 


Currarini, Jackson and Pin~\cite{currarini2009economic} examine friendship patterns in a representative sample of U.S. high schools and build a model of friendship formation based on empirical data. They report that all groups are biased towards same-type friendship relative to demographics, but different homophilic patterns emerge as a function of the group size; while homophily is essentially absent for groups that comprise very small or very large fractions of their school, it is significant for groups that comprise a middle-ranged fraction.
In~\cite{fu2012evolution} it is also claimed that the majority group has greater tendency to homophily. In contrast, we have presented independence between the size of the group and the mixing pattern. Namely, the majority-minority parameter $r$ does not influence the attained equilibria. This inconsistency can be explained by the different construction of the network (\cite{currarini2009economic} and~\cite{fu2012evolution} assume random matching with biased agreement as mentioned above), or perhaps by
the simplicity of our model and the fact that it involves only two groups.
\iftoggle{paper}{}{\vspace{-11pt}}
\section{Network and Game Model}
\iftoggle{paper}{}{\vspace{-8pt}}
Our network model is an extension of the bi-populated biased preferential attachment (BPA) model \cite{avin2015homophily}. We use this model as the basis to an
{\em evolving heterogeneous network game}. 
We start by describing the network model.


\iftoggle{paper}{}{\vspace{-4pt}}
\subsection{Biased Preferential Attachment Model}
The \emph{biased preferential attachment model\footnote{In fact, here we extend the model of \cite{avin2015homophily} to allow heterophily.}} (BPA) ~\cite{avin2015homophily} is a bi-populated preferential attachment model obtained by applying the classical preferential attachment model 
\cite{barabasi1999emergence} to a bi-populated minority-majority network augmented with homophily. 

\begin{definition}[BPA Model, $\bpa(n, r, \STRTGY)$
]
The model describes a bi-populated random evolving network with red and blue vertices, where $n$ is the total number of nodes, 
$r$ is the arrival rate of the red vertices and $\STRTGY$ is the \emph{mixing matrix}.
Denote the social network at time $t$ by $G_t=(V_t,E_t)$, where $V_t$ and $E_t$, respectively, are the sets of vertices and edges in the network at time $t$, and let $d_t(v)$ denote the degree of vertex $v$ at time $t$. 
The process starts with an arbitrary initial bi-populated (red-blue) connected network $G_0$ with $n_0$ vertices and $m_0$ edges.
For simplicity we hereafter assume that $G_0$ consists of one blue and one red vertex connected
by an edge, but this assumption can be removed.
This initial network evolves in $n$ time steps as follows. 
In every time step $t$, a new vertex $v$ enters the network. 
The arrival rate of the red nodes is denoted by $0 < r <1$, i.e., the new 
vertex $v$ is red with probability $r$ and blue with probability $1-r$. 

In the first stage, $v$ selects a {\em tentative} neighbor $u$ at random by preferential attachment, i.e., with probability proportional to $u$'s degree at time $t$, 
\iftoggle{paper}{}{\vspace{-0.2cm}}
$$\Prb[u \text{ is chosen}] = {d_t(u)}/{\sum_{w\in V_t}d_t(w)}.$$
The second stage employs a $2 \times 2$ stochastic \emph{mixing matrix}, $\STRTGY$, composed of the stochastic \emph{homophily vectors} of each player, $\STRTGY_{\RRED}, \STRTGY_{\BBLUE}$, i.e.,
\iftoggle{paper}{}{\vspace{-2pt}}
$$
\STRTGY = \left(
  \begin{array}{c}
    \STRTGY_{\RRED} \\
    \STRTGY_{\BBLUE} \\
  \end{array}
\right)
= \left(
  \begin{array}{cc}
    \rhoRR & 1-\rhoRR \\
    1-\rhoBB & \rhoBB \\
  \end{array}
\right).
$$
\iftoggle{paper}{}{\vspace{-2pt}}
Letting $\CC \in \{\RED, \BLUE\}$ be $v$'s color, the edge $(v,u)$ is inserted into the graph with probability $\rho_\CC$ when $u$'s color is also $\CC$.
If the colors differ, then the edge is inserted with probability $1-\rho_\CC$.
If the edge is rejected (i.e., is not inserted into the graph), then the two-stage procedure is restarted. 
This process is repeated until some edge $\{v,u\}$ has been inserted. Thus in each time step, one new vertex and one new edge are added to the existing graph. 
\end{definition}

Note that the mixing matrix $\STRTGY$ describes the degree of {\em segregation}
(incorporated by using rejection sampling) of the system. In particular, using the {\em perfect homophily} matrix
\def\HOMATRIX{\STRH =\left(
  \begin{array}{c}
    \H_{\RRED} \\
    \H_{\BBLUE} \\
  \end{array}
\right)=
\left(
  \begin{array}{cc}
    1 & 0 \\
    0 & 1 \\
  \end{array}
\right),}
\iftoggle{paper}{$$\HOMATRIX$$}{$\HOMATRIX$}
all added edges connect vertex pairs of the same color. 
%
At the other extreme, using the perfect {\em heterophily} matrix 
\def\HEMATRIX{\STRT 
=
\left(
  \begin{array}{c}
    \T_{\RRED} \\
    \T_{\BBLUE} \\
  \end{array}
\right)
= 
\left(
  \begin{array}{cc}
    0 & 1 \\
    1 & 0 \\
  \end{array}
\right),}
\iftoggle{paper}{$$\HEMATRIX$$}{$\HEMATRIX$}
all added edges connect vertex pairs with different color.
Similarly, using the {\em unbiased} strategy matrix
\def\UMATRIX{\STRU 
=
\left(
  \begin{array}{c}
    \U_{\RRED} \\
    \U_{\BBLUE} \\
  \end{array}
\right)
= 
\left(
  \begin{array}{cc}
    .5 & .5 \\
    .5 & .5 \\
  \end{array}
\right),}
\iftoggle{paper}{$$\UMATRIX$$}{$\UMATRIX$}
edges are connected independently of the node colors. For intermediate
values $0<\rhoRR,\rhoBB<1$, the players show a tendency to favor one kind of interaction over another. When $\rhoRR,\rhoBB>0.5$, the players tend to be
homophilic, and when $\rhoRR,\rhoBB<0.5$, the players tend to be heterophilic. Figure \ref{fig:model} presents three examples of parameter settings for the BPA model on a 200-vertex bi-populated social network with $r=0.3$ ($30\%$ red nodes), using $\STRH, \STRT$ and $\STRU$. 

\iftoggle{paper}{}{\vspace{-11pt}}
\subsection{Evolving Heterogeneous Network Games}
We now define the {\em evolving heterogeneous} $\eh\left(t,r,\STRTGY,\gamma\right)$ 
network game ($\eh$ {\em game}, for short) between the two sub-populations.
The game is played between two players, the red player $\RED$ and the blue player $\BLUE$. (Note that we occasionally use $\RED$ and $\BLUE$ to denote either the \textit{color}, the corresponding \textit{set} of nodes,
or the corresponding 
\textit{player}. The exact meaning will be clear from the context.)

Assume $r$ and $G_0$ are given to the players. Each player $\XX\in\{\RED,\BLUE\}$ can now choose its \emph{strategy vector} as a mixing vector $\STRTGY_{\XX}$ in the mixing matrix $\STRTGY$. Then the network evolves according the biased preferential attachment model $\bpa(t,r,\STRTGY)$.

Let $n_t(\RED)$ and $n_t(\BLUE)$, respectively, denote the number of red and blue nodes at time $t>0$, where $n_t = n_t(\RED)+n_t(\BLUE)=n_0+t$. Denote by $d_t(\RED)$ (respectively, $d_t(\BLUE)$) the sum of degrees of the red (resp., blue) vertices present in the system at time $t\ge 0$. Altogether, the number of edges in the network at time $t$ is $m_t=m_0+t$, where $d_t(\RED)+d_t(\BLUE)=2m_t$. 

Let $C(G_t)$ denote the \textit{cut} of the graph $G_t$ defined by the red-blue partition of $V_t$, i.e., the set of edges that have one endpoint in $\RED$ and the other in $\BLUE$. Formally,
\iftoggle{paper}
{$$C(G_t) = \left\{\left(u,v\right)\in E_t \mid u\in\RED,v\in\BLUE\right\}.$$}
{\\ \centerline{$C(G_t) = \left\{\left(u,v\right)\in E_t \mid u\in\RED,v\in\BLUE\right\}.$}}
Let $\phi(G_t)=\card{C(G_t)}$ denote the size of the cut.

In our game, the \emph{payoff} of each player is a combination of two quantities: the {\em total power} of its sub-population (namely, its expected sum of degrees), and the expected {\em cut size} $\phi(G)$.
Observe that these quantities pull in opposite directions, hence they are balanced using a parameter $0 \leq \gamma \leq 1$ that will serve as a {\em weighting factor} for the {\em utility function} of the game. The parameter $\gamma$ can be viewed as set by a regulator to enforce cooperation between sub-populations.
Formally, the payoffs (utilities) of the
players $\RED$ and $\BLUE$ at time $t$ are
\iftoggle{paper}{}{\vspace{-0.3cm}}
\begin{eqnarray*}
U_t^\gamma(\RED) ~=~ \gamma\frac{d_t(\RED)}{d_t}+(1-\gamma)\frac{\phi_t}{2 m_t}= \frac{1}{d_t}\Big(\gamma d_t(\RED)+(1-\gamma)\phi_t\Big), \\
U_t^\gamma(\BLUE) ~=~ \gamma\frac{d_t(\BLUE)}{d_t}+(1-\gamma)\frac{\phi_t}{2 m_t}=\frac{1}{d_t}\Big(\gamma d_t(\BLUE)+(1-\gamma)\phi_t\Big).
\end{eqnarray*}

A strategy profile $\STRTGY$ is a \textit{Nash equilibrium} for the game $\eh\left(t,r,\STRTGY,\gamma\right)$ if no player $\XX\in\{\RED,\BLUE\}$ can do better by unilaterally changing its own strategy $\STRTGY_{\XX}$.
A Nash equilibrium for the game $\eh\left(t,r,\STRTGY,\gamma\right)$ is \textit{stable} if a small change in $\STRTGY$ for one player leads to a situation where two conditions hold:
$(i)$ the player who did not change has no better strategy in the new circumstance, and
$(ii)$ the player who did change is now playing with a strictly worse strategy.
If both conditions are met, then the player who changed its $\STRTGY$ will return immediately to the Nash equilibrium, hence the equilibrium is {\em stable}. If condition $(i)$ does not hold (but condition $(ii)$ does), then the equilibrium is \textit{unstable}. 

\iftoggle{paper}{}{\vspace{-11pt}}
\section{Degree Maximization Game}
\label{sec:basic}
\iftoggle{paper}{}{\vspace{-8pt}}
Before studying the behavior of the general evolving heterogeneous network game, let us consider the solution of the game in the basic case where $\gamma=1$ for every $t$, i.e., each player's utility depends only on the expected sum of degrees.

\subsubsection*{An urn process.}
The biased preferential attachment $\bpa(n, r,\STRTGY)$ process can also be interpreted as a Polya's urn process, where each new edge added to the graph corresponds to two new balls added to the urn, one for each endpoint, and the balls are colored by the color of the corresponding vertices. In this interpretation, a time step of the original evolving network process corresponds to the arrival of a new ball $x$ (which is red with probability $r$ and blue with probability $1-r$), and in the ensuing procedure, we choose an existing ball $y$ from the urn uniformly at random; now, if $x$ is of the same (respectively, different) color $\CC\in{\RED,\BLUE}$ as $y$, then with probability $\rho_\CC$ (resp., $1-\rho_\CC$) we add to the urn both $x$ and a second copy of $y$ (corresponding to the two endpoints of the added edge), 
and with probability $1-\rho_\CC$ (resp., $\rho_\CC$) we reject the choice of $y$ and repeat the experiment, i.e., choose another existing ball $y'$ from the urn uniformly at random. This is repeated until the choice of $y$ is not rejected. 
Hence the arrival of each new ball $x$ results in the addition of exactly two new balls to the urn, namely, $x$ and a copy of some existing ball $y$.

The key observation is that to analyze the expected fraction of the red balls in the urn at time $t$, there is no need to keep track of the degrees of individual vertices in the corresponding process of evolving network; the sum of degrees of all red vertices, $d_t(\RED)$, is exactly the number of red  balls in the urn. 
Noting that exactly two balls join the system in each time step, we have 
\iftoggle{paper}
{$$d_t(\RED)+d_t(\BLUE)~=~ d_t ~=~ 2t+n_0=2(t+1).$$}
{\\ \centerline{$d_t(\RED)+d_t(\BLUE)~=~ d_t ~=~ 2t+n_0=2(t+1)$.}}
Note that while $d_t(\RED)$ and $d_t(\BLUE)$ are random variables, $d_t$ is not.

\subsubsection*{Convergence of expectations.}
Let $\alpha_t = {d_t(\RED)}/{d_t}$ be a random variable denoting the fraction of red balls in the system at time $t$. 
Given the mixing matrix $\STRTGY$, we claim that 
the process will converge to a ratio of $\alpha$ red balls in the system (as a function of $\STRTGY$).
More formally, we claim that, regardless of the starting condition, there exists a limit
$\alpha ~=~ \lim_{t\rightarrow \infty} \Expct[\alpha_t]~.$


\begin{lemma}
\label{lem:Expectation}
{~~} 
$\displaystyle \Expct\left[\alpha_{t+1} \mid \alpha_{t}\right]=\alpha_{t}+\frac{F\left(\alpha_{t}\right)-\alpha_{t}}{t+2}$~, where \\ 
$$F(x)=\frac{1}{2}\left(1 + \frac{\rhoBB (-1 + r) (-1 + \alpha)}{-\alpha + \rhoBB (-1 + 2 \alpha)} + \frac{r \rhoRR \alpha}{1 - \alpha + \rhoRR (-1 + 2 \alpha)}\right).$$
\end{lemma}

\iftoggle{paper}{%
\begin{proof}
Given that the new vertex at time $t+1$ is blue, the probability $\pBB$ that it attaches to a blue vertex satisfies 
$$
\pBB(\at)=(1-\alpha_t)\rhoBB+\alpha_t\rhoBB \pBB(\at)+(1-\alpha_t)(1-\rhoBB) \pBB(\at),
$$
hence 
$$
\pBB(\at)=\frac{\rhoBB-\rhoBB \alpha_t}{\rhoBB+\alpha_t-2\rhoBB \alpha_t}
$$
Similarly, given that the new vertex at time $t+1$ is red, the probability $P_{\RRED\RRED}$ that it attaches to a red vertex satisfies 
$$
\pRR(\at)=\alpha_t\rhoRR+\alpha_t(1-\rhoRR)\pRR(\at)+(1-\alpha_t)\rhoRR \pRR(\at),
$$
hence
$$
\pRR(\at)=\frac{\rhoRR \alpha_t}{1-\alpha_t +\rhoRR(2\alpha_t-1)}.
$$
We later express $\pBB$ and $\pRR$ as a function of $\alpha_t$, i.e., 
\begin{eqnarray} 
\pBB(x)&=&\frac{\rhoBB-\rhoBB x}{\rhoBB+x-2\rhoBB x}~,
\nonumber \\
\pRR(x)&=&\frac{\rhoRR x}{1-x +\rhoRR(1-2x)}~.
\label{eq:pBBpRR}
\end{eqnarray} 
In each step the sum of the degrees increases by $2$, so $d_{t+1} = d_t + 2$. We start from an arbitrary ratio $\alpha_0 = d_0(R)/d_0$. Let $\NN_t(x)$ be a random variable denoting the number of \emph{new} red balls at time $t$ assuming $\alpha_t=x$. 
Then
\begin{equation} 
\label{eq:Xt}
\begin{split}
\NN_{t+1}(x) &= \begin{cases}
  0, & \text{with probability } (1-r)\pBB(x) \\ 
     & \text{(a blue ball entered and chose a blue ball),} \\
  2, & \text{with probability } r\pRR(x) \\
     & \text{(a red ball entered and chose a red ball),} \\
  1, & \text{with the remaining probability} \\
     & \text{(a blue ball chose a red ball or vice versa),}
\end{cases}
\end{split}
\end{equation}
and 
$$d_{t}(R)=d_0(R)+ \sum_{i=1}^t \NN_t(\alpha_{i-1}).$$
We now define
$\cE_t=\Expct\left[\NN_{t+1}(\alpha_{t})\right]$ and calculate it to be
\begin{equation*} 
\begin{split}
\cE_t &=\Expct\left[d_{t+1}(R)-d_{t}(R) \mid \alpha_{t}\right] \\
&=1\cdot \left((1-(1-r)\pBB(\alpha_t)-r\pRR(\alpha_t)\right) + 2\cdot r\pRR(\alpha_t)\\
&=1-(1-r)\pBB(\alpha_t)+r\pRR(\alpha_t)\\
&=1-(1-r)\frac{\rhoBB-\rhoBB \alpha_t}{\rhoBB+\alpha_t-2\rhoBB \alpha_t}+r\frac{\rhoRR \alpha_t}{1-\alpha_t +\rhoRR(1-2\alpha_t)}\\
&=2F(\alpha_t).
\end{split}
\end{equation*}
Substituting $d_{t+1}(R)=2(t+2)\alpha_{t+1}$ and $d_{t}(R)=2(t+1)\alpha_{t}$ and rewriting yields the lemma.
\qed\end{proof}
}

\iftoggle{paper}{}{\vspace{-0.5cm}}
\begin{lemma}
\label{lem:Fprop}
The function $F(x)$ has the following properties:
\begin{enumerate}
	\item $F(x)$ is monotonically increasing.
 	\item $F(x)$ has exactly one fixed point, $\alpha \in \left[0,1\right]$.
 	\item The image of the unit interval by $F(x)$ is contained in the unit interval: \\ $F\left(\left[0,1\right]\right) =\left[\frac{r}{2},\frac{1+r}{2}\right]\subset\left[0,1\right]$.
 	\item If $x < \alpha$ then $x < F(x) < \alpha$ and if $x > \alpha$ then $x > F(x) > \alpha$.
\end{enumerate}
\end{lemma}
\iftoggle{paper}{%
\begin{proof}
For the first property, observe that
\begin{equation}
\frac{\partial F}{\partial x}  ~=~ 
\frac{1}{2}\left(\frac{(-1 +\rhoBB) \rhoBB (-1 +r)}{(\rhoBB +x -2\rhoBB x)^2} 
- \frac{r (-1 + \rhoRR) \rhoRR}{(-1 + \rhoRR + x - 2 \rhoRR x)^2}\right) 
~>~ 0 
\label{eq:dF}
\end{equation}
for every $x,\rhoRR,\rhoBB\in \left[0,1\right]$ and $r\in \left(0,1\right)$.

For the second property, we define the function $G(x) = F(x) - x$. The roots of $G(x)$ correspond to the fixpoints of $F(x)$ so we will show that $G(x)$ has exactly one real
root in the interval $\left[0,1\right]$. We arrange $G(x)$ as $G(x)=\frac{Q(x)}{W(x)}$ where 
$$W(x)=2 (-x + \rhoBB (-1 + 2 x)) (1 - x + \rhoRR (-1 + 2 x)).$$
Since the denominator $W(x)$ is positive for each $r,x,\rhoBB,\rhoRR\in\left[0,1\right]$, it is enough to show that the numerator $Q(x)$ has exactly one real
root in the interval $\left[0,1\right]$ as shown in Lemma~\ref{clm:uniqueFix} below.

The third property follows from the fact that the function $F(x)$ is strictly monotonically increasing and by evaluating the function $F(x)$ for the two extreme values $F(0)=r/2$, and $F(1)=(1+r)/2$.

Finally, the fourth property follows from the fact that the function is strictly monotonically increasing, that there is only one fixed point and that $F(x)$ maps $\left[0,1\right]$ inside $\left[0,1\right]$.
\qed\end{proof}

\begin{lemma}
\label{clm:uniqueFix}
The polynomial
\begin{eqnarray*}
  Q(x) = 2 (-1 + 2 \rhoBB) (-1 + 2 \rhoRR) x^3 +(-3 + 7 \rhoBB + \rhoBB r + 4 \rhoRR - 10 \rhoBB \rhoRR + r \rhoRR - 4 \rhoBB r \rhoRR) x^2\\
  + (1 - 3 \rhoBB - 2 \rhoBB r - \rhoRR + 3 \rhoBB \rhoRR + 4 \rhoBB r \rhoRR) x + \rhoBB r - \rhoBB r \rhoRR
\end{eqnarray*}
has a unique root in $\left[0,1\right]$.
\end{lemma}

\begin{proof}

In what follows, we employ Sturm's Theorem (to be explained next) in order to bound the number of distinct real roots of $Q(x)$. 

Consider some degree n polynomial $P(x) =a_n x^n+ ... +a_1 x + a_0$ over the reals. The Sturm sequence of $P(x)$ is a sequence of polynomials denoted by $p_0(x),p_1(x),... ,p_m(x)$, where $p_0(x) = P(x)$, $p_1(x) = dP(x)/dx$, and
$p_i(x) = remainder(p_{i-2}(x)/p_{i-1}(x))$ for $i > 1$. This recursive definition terminates at step $m$ such
that $remainder(p_{m-1}(x)/p_m(x)) = 0$. Since the degree of $p_i(x)$ is at most $n-i$, we conclude that
$m\leq n$. Define $SC_p(t)$ to be the number of sign changes in the sequence $p_0(t),p_1(t),... ,p_m(t)$.
We are now ready to state the following theorem attributed to Jacques Sturm, $1829$ (cf.~\cite{basu2005algorithms}).
\begin{theorem}[Sturm's condition]
Consider two reals $a,b$, where $a < b$ and neither of them is a
root of $P(x)$. Then the number of distinct real roots of $P(x)$ in the interval $(a,b)$ is $SC_p(a)-SC_p(b)$.
\end{theorem}

Let's examine the Sturm sequence of $Q$ in $\left(0,1\right)$ 
for every $\rhoRR,\rhoBB$, checking different $\rho$ ranges as follows.

For $\rhoBB\in(0,\frac{1}{2})$ and $\rhoRR\in(\frac{1}{2},1)$:
To find the number of roots between $0$ and $1$, 
we first evaluate $p_0(x),p_1(x),p_2(x)$ and $p_3(x)$ at $x=1$ and get the sequences of signs of the results: $\left\{-,-,-,-\right\}$, which contains no sign changes. Evaluating $p_0(x),p_1(x),p_2(x)$ and $p_3(x)$ at $x=0$ yields two optional sequences of signs of the results: for $1/2 < \rhoRR < (1 + 2 r)/(1 + 4 r)$ and $(-1 + \rhoRR)/(-3 + 3  \rhoRR - 2 r + 4  \rhoRR r) < 
   \rhoBB < 1/2$ we get
$\left\{+,-,-,-\right\}$. Otherwise, we get the sequences of signs $\left\{+,+,*,-\right\}$ (where $*$ is $+$ or $-$). All of the sequences contain one sign changes, hence, the number of roots of $Q$ between $0$ and $1$ is $1- 0 = 1$ as needed.

For $\rhoBB\in(\frac{1}{2},1)$ and $\rhoRR\in(0,\frac{1}{2})$:
To find the number of roots between $0$ and $1$, we evaluate $p_0(x),p_1(x),p_2(x)$ and $p_3(x)$ at $x=0$ and get the sequence of signs: $\left\{+,-,+,-\right\}$ which contains three sign changes. The same procedure for $x=1$ gives for $1/2 < \rhoBB < (-3 + 2 r)/(-5 + 4 r)$ and $(1 - \rhoBB)/(5 - 7 \rhoBB - 2 r + 4 \rhoBB r) <  \rhoRR < 1/2$ the sign sequences: $\left\{-,-,+,-\right\}$. Otherwise, we get the sequences $\left\{-,+,*,-\right\}$ (where $*$ is $+$ or $-$) . Since all of these contain two sign changes, we get that the number of roots of $Q$ between $0$ and $1$ is $3-2 = 1$ as needed.

For $\rhoRR,\rhoBB\in(0,\frac{1}{2})$ or $\rhoRR,\rhoBB\in(\frac{1}{2},1)$ we get that $Q(1)<1$ and $Q(0)>0$. Observe that $Q(x)=\infty$ when $x\to\infty$ and $Q(x)=-\infty$ when $x\to-\infty$. This implies that there are one or three roots in both intervals $(-\infty,0)$ and $(1,\infty)$. Knowing that $Q(x)$ has exactly three roots concludes the claim that $G$ has exactly one root in $[0, 1]$.

Finally, when either $\rhoRR=1$ and $0\leq\rhoBB<1$, or $\rhoBB=0$ and $0<\rhoRR<1$, there is a root of $M$ at $x=0$, hence these cases must be dealt with separately. Another special case occurs when $\rhoRR,\rhoBB=\frac{1}{2}$. For each of these special cases
we explicitly solve the equation $M(x)=0$ and show that there is a unique root at $\left(0,1\right)$.
Lemma \ref{clm:uniqueFix} follows.
\qed\end{proof}
}

Assume w.l.o.g. that $\alpha_t < \alpha$. By Lemma~\ref{lem:Fprop} $\alpha_t < F(\alpha_t) < \alpha$, so by Lemma~\ref{lem:Expectation} 
$\alpha_t < \Expct\left[\alpha_{t+1} \mid \alpha_{t}\right] < \alpha$. Taking expectations, we get that
$\Expct[\alpha_t] < \Expct[\alpha_{t+1}] < \Expct[\alpha]=\alpha$.
We have thus shown that the expected value of $\alpha_t$ converges to the fixed point $\alpha$ of $F(x)$. We have thus established the following.

\begin{theorem}
\label{thm:exp_conv}
Given the rate $r$ of red nodes and the mixing matrix $\STRTGY$, for any initial graph, as $t$ tends to infinity, the expected fraction of red balls, $\Expct[\alpha_t]$, converges to the unique real $\alpha \in (0,1)$ satisfying  the equation $F(\alpha)=\alpha$, or
\begin{equation*}
\begin{split}
2\alpha ~=~1 + \frac{\rhoBB (-1 + r) (-1 + \alpha)}{-\alpha + \rhoBB (-1 + 2 \alpha)} + \frac{r \rhoRR \alpha}{1 - \alpha + \rhoRR (-1 + 2 \alpha)}~.
\end{split}
\end{equation*}
\end{theorem}
Hence the limit $\alpha$ is the solution of the cubic equation 
\iftoggle{paper}{}{\vspace{-6pt}}
\begin{eqnarray*}
&  (2 - 4 \rhoBB - 4 \rhoRR + 8 \rhoBB \rhoRR) \alpha^3  + (-3 + 7 \rhoBB + \rhoBB r + 4 \rhoRR - 10 \rhoBB \rhoRR + r \rhoRR - 4 \rhoBB r \rhoRR) \alpha^2\\
  &\mbox{\hskip 1cm} +  (1 - 3 \rhoBB - 2 \rhoBB r - \rhoRR + 3 \rhoBB \rhoRR + 4 \rhoBB r \rhoRR) \alpha  +  \rhoBB r - \rhoBB r \rhoRR ~=~0.
\end{eqnarray*}
Note that this limit is independent of the initial values $d_0$ and $\alpha_0$ of the system.


\subsubsection*{Existence of a Nash Equilibrium.}
Having shown that for any given strategy profile $\STRTGY$ the expected fraction of red node degrees converges to $\alpha$, we examine the influence of the different strategies on the utility functions.

\begin{lemma}
\label{lem:alpha_rho_relation}
The limit $\alpha$ and $\Expct[\alpha_t]$ are monotone in the mixing matrix entries, i.e., both increase with increasing $\rhoRR$ and decrease with increasing $\rhoBB$.
\end{lemma}

\iftoggle{paper}{%
\begin{proof}
We show (strict) monotonicity in $\rhoRR$; a similar proof can be obtained for  $\rhoBB$. 
Consider two urn processes $\PROCESS$ and $\PROCESS'$ corresponding to the games $\bpa(n,G_0,r,\STRTGY)$ and $\bpa(n,G_0,r,\STRTGY')$, where
$$
\STRTGY = \left(
  \begin{array}{cc}
    \rhoRR & 1-\rhoRR \\
    1-\rhoBB & \rhoBB \\
  \end{array}
\right)
~~~~\mbox{ and }~~~~
\STRTGY' = \left(
  \begin{array}{cc}
    \rhoRR+\epsilon' & 1-(\rhoRR+\epsilon') \\
    1-\rhoBB & \rhoBB \\
  \end{array}\right)$$ 
for some $\epsilon'>0$.
Denote by $\alpha_t=\frac{d_t(\RED)}{d_t}$ and $\alpha_t'=\frac{d'_t(\RED)}{d_t}$ the fraction of red balls at time $t$ in $\PROCESS$ and $\PROCESS'$ respectively.
Let $\alpha=\lim_{t\rightarrow \infty} \Expct[\alpha_t]$ and $\alpha'=\lim_{t\rightarrow \infty} \Expct[\alpha'_t]$. 

In order to prove the first part of the lemma 
(i.e., the claim on the limit $\alpha$)
we show that $\alpha<\alpha'$.
Let $F(x)$ and $F'(x)$ be the functions defined for
the process $\PROCESS$ and $\PROCESS'$ respectively.
Observe that $\partial F/\partial\rhoRR>0$ for each $\rho,r\in[0,1]$ and $x\in(0,1)$, so $F(x)<F'(x)$ for every $x\in(0,1)$. Note that $F(\alpha)=\alpha$ and $F'(\alpha')=\alpha'$ are the unique fixed points of $F(x)$ and $F'(x)$, respectively, hence $\alpha=F(\alpha)<F'(\alpha')=\alpha'$ as required.


The proof of the second part of the lemma 
(i.e., the claim on $\Expct[\alpha_t]$)
uses stochastic domination (see cf.~\cite{shaked2007stochastic}). We give the formal definition and a basic theorem that we use.
\begin{definition}[Stochastic domination]
Let $X$ and $Y$ be two random variables, not necessarily
on the same probability space. The random variable $X$ is {\em stochastically smaller than} $Y$, denoted $X\preceq Y$, if $\Prb[X> z]\leq\Prb[Y> z]$ for every $z\in\Rnum$. If additionally $\Prb[X>z]<\Prb[Y>Z]$ for some $z$, then $X$ is {\em stochastically strictly less than} $Y$, denoted $X\prec Y$.
\end{definition}
\begin{theorem}[stochastic order]
\label{thm:stocDom}
Let $X$ and $Y$ be two random variables, not necessarily on the same probability space.
\begin{enumerate}
	\item Suppose $X\prec Y$. Then $\Expct[U(X)]<\Expct[U(Y)]$ for any strictly increasing continuous utility function $U$.
	\item Suppose $X_1\prec Y_1$ and $X_2\prec Y_2$, for four random variables $X_1, Y_1,X_2$ and $Y_2$. Then $aX_1 + bY_1 \prec aX_2 + bY_2$ for any two constants $a,b>0$.
\end{enumerate}
\end{theorem}

Let $\NN_{t+1}(x)$ (respectively, $\NN'_{t+1}(x)$) be a random variable denoting the number of new red balls at time $t+1$ in $\PROCESS$ (resp., $\PROCESS'$) assuming $\alpha_t=x$ (resp., $\alpha'_t=x$). 

\begin{lemma}
\label{clm:x}
$\NN_{t+1}(x) \prec \NN'_{t+1}(x)$ for any $0<x<1$ and integer $t\ge 0$.
\end{lemma}

\begin{proof}
By Eq. (\ref{eq:pBBpRR}) and (\ref{eq:Xt}),
\begin{equation*} 
\begin{split}
\Prb[\NN_{t+1}(x)=0]= \Prb[\NN'_{t+1}(x)=0],\\
\Prb[\NN_{t+1}(x)=2]< \Prb[\NN'_{t+1}(x)=2].
\end{split}
\end{equation*}
Hence $\Prb[\NN_{t+1}(x)>z]\leq \Prb[\NN'_{t+1}(x)> z]$ for every $z\in\Rnum$ and $\Prb[\NN_1(\alpha_0)> 1] < \Prb[\NN_1'(\alpha_0)> 1]$, yielding $\NN_{t+1}(x) \prec \NN'_{t+1}(x)$.
\qed\end{proof}

\begin{lemma}
\label{clm:alphaXrelation}
For $t\ge 0$, if $\alpha_t\prec \alpha'_t$ then $\NN_{t+1}(\alpha_t) \prec \NN'_{t+1}(\alpha'_t)$.
\end{lemma}
\begin{proof}
We would like to show that for every $z$,
$$\Prb[\NN_{t+1}(\alpha_t)> z] \leq \Prb[\NN'_{t+1}(\alpha'_t)> z].$$
Denoting expectation according to the r.v. $Z$ by $\Expct_{Z}[\cdot]$, we have
\begin{eqnarray*}
\Prb[\NN_{t+1}(\alpha_t)> z] &=& \Expct_{\alpha_t} \left[ \Prb[\NN_{t+1}(\alpha_t)> z] \right] \notag 
~\le~ \Expct_{\alpha_t} \left[ \Prb[\NN'_{t+1}(\alpha_t)> z] \right] 
\\ 
&\le& \Expct_{\alpha'_t} \left[ \Prb[\NN'_{t+1}(\alpha'_t)> z] \right] = \Prb[\NN'_{t+1}(\alpha'_t)> z], 
\end{eqnarray*}
where the first inequality 
follows from Lemma~\ref{clm:x}, which shows that $\NN_{t+1}(x) \preceq \NN'_{t+1}(x)$, and 
the second is by Theorem~\ref{thm:stocDom}(1), noting that $\Prb[\NN'_{t+1}(x)> z]$ is monotone in $x$.

To show strictness (i.e., $\NN_{t+1}(\alpha_t) \prec \NN'_{t+1}(\alpha'_t)$) we consider $z=1$ and show that
$\Prb[\NN_{t+1}(\alpha_t)> 1] < \Prb[\NN'_{t+1}(\alpha'_t)> 1]$.
\qed\end{proof}

\begin{lemma}
\label{clm:alphaVsalpha'}
$\alpha_t\prec\alpha'_t$ for $t\ge 0$.
\end{lemma}

\begin{proof}
Note that 
$$d_{t}(\RED)= d_0(\RED) + \displaystyle\sum_{i=1}^{t} \NN_{i}(\alpha_{i-1}),$$ 
$$d'_{t}(\RED)=d'_0(\RED) + \displaystyle\sum_{i=1}^{t} \NN'_{i}(\alpha'_{i-1}).$$ 
We prove the claim by induction, over $t$.

\noindent\textit{Induction basis.}
$d_0(\RED)=d'_0(\RED)=c_\RRED$ for some constant
$c_\RRED>0$. Then $\alpha_0=\frac{c_\RRED}{d_0}=\alpha'_0$. It follows that
\begin{equation*} 
	\begin{split}
	\Prb[\NN_1(\alpha_0)=0] &=\frac{(1-r)\rhoBB(\alpha_0-1)}{\rhoBB(\alpha_0-1)-(1-\rhoBB)\alpha_0} 
=\Prb[\NN'_1(\alpha'_0)=0] 
	\end{split}
\end{equation*}
and
\begin{equation*} 
\begin{split}
\Prb[\NN_1(\alpha_0)=2] 
~=~ \frac{r\rhoRR\alpha_0}{(1-\rhoRR)(1-\alpha_0)+\rhoRR\alpha_0} 
~<~ \frac{r(\rhoRR+\epsilon)\alpha_0}{(1-\rhoRR)(1-\alpha_0)+(\rhoRR+\epsilon)\alpha_0}
~=~ \Prb[\NN'_1(\alpha'_0)=2],
\end{split}
\end{equation*}
hence $\Prb[\NN_1(\alpha_0)> z]\leq \Prb[\NN_1(\alpha'_0)> z]$ for every $z\in\Rnum$ and $\Prb[\NN_1(\alpha_0)> 1] < \Prb[\NN_1(\alpha'_0)> 1]$, yielding $\NN_1(\alpha_0)\prec \NN'_1(\alpha'_0)$.

\noindent\textit{Induction step}.
Suppose that $\alpha_t\prec \alpha'_t$ holds. By Lemma~\ref{clm:alphaXrelation}, $\NN_{t+1}(\alpha_t)\prec \NN'_{t+1}(\alpha'_t)$. Hence
\begin{equation*} 
\begin{split}
 d_{t+1}(\RED)= d_t(\RED)+\NN_{t+1}(\alpha_t) \prec d'_t(\RED)+\NN'_{t+1}(\alpha'_t)=d_{t+1}(\RED),
\end{split}
\end{equation*} 
where $d_t(\RED) \prec d'_t(\RED)$ by the induction assumption. Note we also used Theorem~\ref{thm:stocDom}(2).
This implies $\alpha_{t+1}\prec \alpha'_{t+1}$ as needed.
\qed\end{proof}

By Theorem~\ref{thm:stocDom} we get $\Expct[\alpha_t]<\Expct[\alpha'_t]$, which complete the proof of the second part of Lemma \ref{lem:alpha_rho_relation}.
\qed\end{proof}
}

Given the utility functions $U^1_t(\RED)=d_t(\RED)$ and $U^1_t(\BLUE)=d_t(\BLUE)$, each player can choose its row in the mixing matrix $\STRTGY$.  By Theorem~\ref{thm:exp_conv} we get that
$U^1_{t\to\infty}(\RED)=d_t\alpha$ and $U^1_{t\to\infty}(\BLUE)=d_t(1-\alpha)$. Lemma~\ref{lem:alpha_rho_relation} implies that the red and blue players maximize their utility by increasing $\rhoRR$ and $\rhoBB$, respectively. Hence, the homophily strategy profile $\STRH$ is strictly dominant for both players. The same applies for $t<\infty$.

\begin{theorem}
\label{thm:basic_game}
The homophily strategy profile $\STRH$ is a unique Nash equilibrium for the game $\eh\left(t,r,\STRTGY,\gamma=1\right)$.
\end{theorem}

\iftoggle{paper}{}{\vspace{-15pt}}
\section{Utilitiy Maximization Game}
\iftoggle{paper}{}{\vspace{-8pt}}
The evolving heterogeneous network game $\eh\left(t,r,\STRTGY,\gamma\right)$ for a bi-populated network consists of two contrasting ingredients, the expected sum of degrees $d(\cdot)$ and the cut size $\phi(G)$. The following theorem expresses the impact of these forces on the system as a function of the weighting factor $\gamma$.
\begin{theorem}
\label{thm:general_nash}
Consider the evolving network game $\eh\left(t,r,\STRTGY,\gamma\right)$ for $0 < r < 1$.
\begin{enumerate}
	\item For $\gamma>1/2$, the homophily strategy profile $\STRH$ is a unique Nash equilibrium.
	\item For $\gamma<1/2$, the heterophily strategy profile $\STRT$ is a unique Nash equilibrium.
	\item For $\gamma=1/2$, the only two Nash equilibria are $\STRH$ and $\STRT$. The homophily strategy profile $\STRH$ is a {\em stable} Nash equilibrium, while the heterophily strategy profile $\STRT$ is an {\em unstable} Nash equilibrium. 
\end{enumerate}
\end{theorem}

\iftoggle{paper}{
\begin{proof}

Let $\MM_t(x)$ be a random variable denoting the number of new cut edges at time $t$. We have
\iftoggle{paper}{}{\vspace{-5pt}}
\begin{equation} 
\begin{split}
\MM_{t+1}(x) &= \begin{cases}
  0, & \text{with probability } (1-r)\pBB(x)+ r\pRR(x), \\ 
  1, & \text{with the remaining probability}, \nonumber 
\end{cases}
\end{split}
\end{equation}
and 
$$\phi(G_t)=\phi(G_0)+\sum_{i=1}^t \MM_i(\alpha_{i-1}).$$
Define the \textit{potential function} of the red player, denoted $\dR$, as the expected increment of its utility at step $t$. Then
\begin{eqnarray*} 
\dR &=& \Expct\left[U_{t+1}^\gamma(\RED)-U_{t}^\gamma(\RED) \mid \alpha\right]=\gamma \NN_{t+1}(\alpha)+(1-\gamma)\MM_{t+1}(\alpha)\\
&=&\gamma(1-(1-r)\pBB(\alpha)+ r\pRR(\alpha)) +(1-\gamma)(1-((1-r)\pBB(\alpha)+ r\pRR(\alpha)))\\
&=& 1-(1-r)\pBB(\alpha)+r(2\gamma -1)\pRR(\alpha)~.\\
\end{eqnarray*}
Similar considerations imply that the \textit{potential function} of the blue player is:
\begin{equation*} 
\dB ~=~ 1 - r \pRR(\alpha) + (1-r)(2\gamma -1)\pBB(\alpha).
\end{equation*}

Let's examine the value of the potential functions $\dR$ and $\dB$ for every $\gamma$, checking different $\gamma$ ranges as follows.

\noindent
$\gamma > 1/2$: In this range the value of $\pRR$ contributes \textit{positively} to $\dR$ and \textit{negatively} to $\dB$. Hence, the red player would like to increase $\pRR$. This would be done by increasing $\rhoRR$ as shown in Lemma~\ref{clm:PBB_rho_relation}.
Similarly, in this range
$\pBB$ contributes \textit{positively} to $\dB$ and \textit{negatively} to $\dR$. Hence, the blue player prefers to increase $\rhoBB$. 
It follows
that the homophily strategies $\H_{\RRED}$ and $\H_{\BBLUE}$ are strictly dominant for both players.

Note that this result also holds for the special case where $\gamma=1$, as shown in Theorem~\ref{thm:basic_game}.

\noindent
$\gamma < 1/2$: Here, both $\pRR$ and $\pBB$ provide \textit{negative} contributions to $\dR$ and $\dB$. Therefore, decreasing $\rhoRR$ implies decreasing $\pRR$ but also increasing $\pBB$ (see Lemma~\ref{clm:PBB_rho_relation}). The variation of $\pBB$ is due to the influence of $\rhoRR$ on $1-\alpha$, which is similar to the variation of $\pRR$ due to $\alpha$. However, $\pRR$ is also decreased directly by $\rhoRR$, hence the red player prefers to decrease $\rhoRR$.
Similarly, the blue player would like to decrease $\rhoBB$, which implies that the heterophily strategies $\T_{\RRED}$ and $\T_{\BBLUE}$ are strictly dominant.

Note that this result also holds for the special case where $\gamma=0$. In this case, the utility is based only on the cut $G(\phi)$, so it is clear that the best strategy for both players is to attach to a node of the opposite color as dictated by the heterophily strategy.

\noindent
$\gamma = 1/2$: In this range the potential function value is 
$$\dR=\left(1-(1-r)\pBB\right)=\left(1-(1-r)\frac{\rhoBB(1-\alpha)}{\rhoBB(1-\alpha)+\alpha(1-\rhoBB)}\right).$$
 Although the strategy of the red player, $\rhoRR$, does not appear explicitly in this expression, it appears implicitly in $\alpha$. Setting $\rhoRR=0$ implies $\pBB=0$, yielding $\dR=1$. Similarly, setting $\rhoBB=0$ yields $\dB=1$. Since $1$ is the maximum value of $\dR$ and $\dB$, it follows that the heterophily strategies are dominant for both players, i.e., $\STRT$ is a Nash equilibrium.  

However, as in the case of $\gamma>1/2$, when $\rhoBB> 0$ the red player would minimize $\alpha$ by 
increasing $\rhoRR$ as shown in Lemma~\ref{lem:alpha_rho_relation}. Similarly, when $\rhoRR> 0$ the blue player would increase $\rhoBB$. 
This leads both players to the homophily strategies $\H_{\RRED}$ and $\H_{\BBLUE}$. Thus, 
$\STRT$ is an \textit{unstable} Nash equilibrium and $\STRH$ is a \textit{stable} Nash equilibrium.
\qed\end{proof}

\begin{lemma}
\label{clm:PBB_rho_relation}
$\pRR$ and $\pBB$ are monotone in the entries of the mixing matrix:
\begin{itemize}
\item $\pRR$ increases with increasing $\rhoRR$ and decreases with increasing $\rhoBB$, and
\item $\pBB$ increases with increasing $\rhoBB$ and decreases with increasing $\rhoRR$.
\end{itemize}
\end{lemma}
\begin{proof}
Observing that $\frac{\partial \pRR}{\partial \rhoRR}>0$ and $\frac{\partial \pRR}{\partial \alpha}>0$, and using Lemma~\ref{lem:alpha_rho_relation}, yields the first part of the claim.
Similarly, $\frac{\partial \pBB}{\partial \rhoBB}>0$ and $\frac{\partial \pBB}{\partial \alpha}<0$ yield the second part.
\qed\end{proof}
}
{
\iftoggle{paper}{}{\vspace{-0.2cm}}
\begin{proof}[Sketch of proof]
Given that the new vertex at time $t+1$ is blue, the probability $\pBB$ that it attaches to a blue vertex satisfies 
\iftoggle{paper}
{$$\pBB(\at)=(1-\alpha_t)\rhoBB+\alpha_t\rhoBB \pBB(\at)+(1-\alpha_t)\rhoBB \pBB(\at),$$}
{\\ \centerline{$\pBB(\at)=(1-\alpha_t)\rhoBB+\alpha_t\rhoBB \pBB(\at)+(1-\alpha_t)\rhoBB \pBB(\at),$}}
hence 
$
\pBB(\at)=\frac{\rhoBB-\rhoBB \alpha_t}{\rhoBB+\alpha_t-2\rhoBB \alpha_t}.
$
Similarly, when the new vertex at time $t+1$ is red, the probability 
that it attaches to a red vertex is
$\pRR(\at)=\frac{\rhoRR \alpha_t}{1-\alpha_t +\rhoRR(1-2\alpha_t)}.$

Let $\NN_t(x)$ and $\MM_t(x)$ be random variables denoting, respectively, 
the number of \emph{new} red balls and cut edges at time $t$. 
We have
$d_{t}(R)=d_0(R)+ \sum_{i=1}^t \NN_t(\alpha_{i-1})$ and $\phi(G_t)=\phi(G_0)+\sum_{i=1}^t \MM_i(\alpha_{i-1}).$
Define the \textit{potential function} of the red player, denoted $\dR$, as the expected increment of its utility at step $t$. Then
\iftoggle{paper}{}{\vspace{-4pt}}
\begin{eqnarray*} 
\dR &=& \Expct\left[U_{t+1}^\gamma(\RED)-U_{t}^\gamma(\RED) \mid \alpha\right]=\Expct\left[\gamma \NN_{t+1}(\alpha)+(1-\gamma)\MM_{t+1}(\alpha)\right]\\
&=&\gamma(1-(1-r)\pBB(\alpha)+ r\pRR(\alpha)) +(1-\gamma)(1-((1-r)\pBB(\alpha)+ r\pRR(\alpha)))\\
&=& 1-(1-r)\pBB(\alpha)+r(2\gamma -1)\pRR(\alpha)~.
\end{eqnarray*}
\iftoggle{paper}{}{\vspace{-2pt}}
Similar considerations imply that the potential function of the blue player is 
\iftoggle{paper}
{$$\dB ~=~ 1 - r \pRR(\alpha) + (1-r)(2\gamma -1)\pBB(\alpha).$$}
{\\ \centerline{$\dB ~=~ 1 - r \pRR(\alpha) + (1-r)(2\gamma -1)\pBB(\alpha)$.}}
The theorem follows by inspecting the value of the potential functions $\dR$ and $\dB$ for every $\gamma$ and using Lemma~\ref{lem:alpha_rho_relation} (for the monotonicity of $\pRR(\alpha)$ and $\pBB(\alpha)$ with the entries of the mixing matrix).
\qed\end{proof}
}

\iftoggle{paper}{}{\vspace{-20pt}}
\section{Discussion}
\iftoggle{paper}{}{\vspace{-9pt}}

This work investigates the assortative mixing phenomenon using a game theory perspective. Given some predefined rules related to the probability of connecting to other node, each player is allowed to determine its strategy in order to maximize its payoff.
First we used a utility function that captures degree centrality, and showed that the expected sum of degrees and its limit are monotonically increasing with the homophily tendency. This directly implies that the homophily strategy is the unique Nash equilibrium. In this context, it will be interesting to use different centrality measures (such as PageRank, betweenness, etc.) and examine their influence on the equilibria. 
Next we enhanced the utility function to give positive payoff for both the degree and the cut. 
The results we have presented show a phase transition in the strategy as a function the weight $\gamma$. A small fluctuation in $\gamma$ might cause extreme changes in the preference of the players, i.e., from perfect homophily to perfect heterophily (or vice versa); the intermediate strategies are never in equilibrium.
This result is independent of the fraction of the sub-population size in the population. 
Generalizing the model to more than two sub-populations or reformulating the utility function may shape the strategy function differently. 

An interesting outcome of the above is the possibility that setting a rule (or a law) by a regulator to encourage cooperation between the two sub-populations will play as a remedial strategy to achieve equal opportunities. This observation is remarkable since, in contrast to the usual affirmative action approach, this attitude does not discriminate any individual, but at the same time, it promises a fair representation of the different sub-populations and even a way for breaking the glass ceiling~\cite{avin2015homophily} that some minority sub-populations suffer from. We leave this direction for further work.

\iftoggle{paper}{\clearpage}{\vspace{-16pt}}

\iftoggle{paper}{}{\vspace{-0.4cm}}
\end{document}